\documentclass[a4paper,onecolumn,superscriptaddress,11pt]{quantumarticle}
\pdfoutput=1

\usepackage[utf8]{inputenc}
\usepackage[english]{babel}
\usepackage[T1]{fontenc}
\usepackage{amsmath}
\usepackage{hyperref}
\usepackage[numbers,sort&compress]{natbib}

\usepackage{tikz}

\usepackage{graphicx,epic,eepic,epsfig,amsmath,latexsym,amssymb,verbatim,color}

\usepackage{theorem}

\newtheorem{definition}{Definition}
\newtheorem{proposition}[definition]{Proposition}
\newtheorem{lemma}[definition]{Lemma}

\newtheorem{theorem}[definition]{Theorem}

\def\squareforqed{\hbox{\rlap{$\sqcap$}$\sqcup$}}
\def\qed{\ifmmode\squareforqed\else{\unskip\nobreak\hfil
\penalty50\hskip1em\null\nobreak\hfil\squareforqed
\parfillskip=0pt\finalhyphendemerits=0\endgraf}\fi}
\def\endenv{\ifmmode\;\else{\unskip\nobreak\hfil
\penalty50\hskip1em\null\nobreak\hfil\;
\parfillskip=0pt\finalhyphendemerits=0\endgraf}\fi}
\newenvironment{proof}{\noindent \textbf{{Proof~} }}{\qed}

\newenvironment{example}{\noindent \textbf{{Example~}}}{}

\mathchardef\ordinarycolon\mathcode`\:
\mathcode`\:=\string"8000
\def\vcentcolon{\mathrel{\mathop\ordinarycolon}}
\begingroup \catcode`\:=\active
  \lowercase{\endgroup
  \let :\vcentcolon
  }

\newcommand{\nc}{\newcommand}
\nc{\rnc}{\renewcommand}
\nc{\beg}{\begin{equation}}
\nc{\eeq}{{\end{equation}}}
\nc{\beqa}{\begin{eqnarray}}
\nc{\eeqa}{\end{eqnarray}}
\nc{\lbar}[1]{\overline{#1}}
\nc{\bra}[1]{\langle#1|}
\nc{\ket}[1]{|#1\rangle}
\nc{\ketbra}[2]{|#1\rangle\!\langle#2|}
\nc{\braket}[2]{\langle#1|#2\rangle}

\nc{\proj}[1]{| #1\rangle\!\langle #1 |}
\nc{\avg}[1]{\langle#1\rangle}
\nc{\Rank}{\operatorname{Rank}}
\nc{\smfrac}[2]{\mbox{$\frac{#1}{#2}$}}
\nc{\tr}{\operatorname{Tr}}
\nc{\ox}{\otimes}
\nc{\dg}{\dagger}
\nc{\dn}{\downarrow}
\nc{\cA}{{\cal A}}
\nc{\cB}{{\cal B}}
\nc{\cC}{{\cal C}}
\nc{\cD}{{\cal D}}
\nc{\cE}{{\cal E}}
\nc{\cF}{{\cal F}}
\nc{\cG}{{\cal G}}
\nc{\cH}{{\cal H}}
\nc{\cI}{{\cal I}}
\nc{\cJ}{{\cal J}}
\nc{\cK}{{\cal K}}
\nc{\cL}{{\cal L}}
\nc{\cM}{{\cal M}}
\nc{\cN}{{\cal N}}
\nc{\cO}{{\cal O}}
\nc{\cP}{{\cal P}}
\nc{\cQ}{{\cal Q}}
\nc{\cR}{{\cal R}}
\nc{\cS}{{\cal S}}
\nc{\cT}{{\cal T}}
\nc{\cX}{{\cal X}}
\nc{\cZ}{{\cal Z}}
\nc{\csupp}{{\operatorname{csupp}}}
\nc{\qsupp}{{\operatorname{qsupp}}}
\nc{\var}{{\operatorname{var}}}
\nc{\rar}{\rightarrow}
\nc{\lrar}{\longrightarrow}
\nc{\polylog}{{\operatorname{polylog}}}
\nc{\1}{{\openone}}
\nc{\wt}{{\operatorname{wt}}}
\nc{\av}[1]{{\left\langle {#1} \right\rangle}}

\def\a{\alpha}

\def\x{\xi}

\def\o{\omega}

\def\D{\Delta}

\def\U{\Upsilon}

\nc{\RR}{{{\mathbb R}}}
\nc{\CC}{{{\mathbb C}}}
\nc{\FF}{{{\mathbb F}}}
\nc{\NN}{{{\mathbb N}}}
\nc{\ZZ}{{{\mathbb Z}}}
\nc{\PP}{{{\mathbb P}}}
\nc{\QQ}{{{\mathbb Q}}}
\nc{\UU}{{{\mathbb U}}}
\nc{\EE}{{{\mathbb E}}}
\nc{\id}{{\operatorname{id}}}

\nc{\CHSH}{{\operatorname{CHSH}}}

\nc{\be}{\begin{equation}}
\nc{\ee}{{\end{equation}}}
\nc{\bea}{\begin{eqnarray}}
\nc{\eea}{\end{eqnarray}}
\nc{\<}{\langle}
\rnc{\>}{\rangle}
\nc{\Hom}[2]{\mbox{Hom}(\CC^{#1},\CC^{#2})}
\nc{\rU}{\mbox{U}}

\nc{\ob}[1]{#1}

\nc{\SEP}{{\text{SEP}}}
\nc{\NS}{{\text{NS}}}
\nc{\LOCC}{{\text{LOCC}}}
\nc{\PPT}{{\text{PPT}}}
\nc{\EXT}{{\text{EXT}}}
\nc{\Sym}{{\operatorname{Sym}}}

\nc{\ERLO}{{E_{\text{r,LO}}}}
\nc{\ERLOCC}{{E_{\text{r,LOCC}}}}
\nc{\ERPPT}{{E_{\text{r,PPT}}}}
\nc{\ERLOCCinfty}{{E^{\infty}_{\text{r,LOCC}}}}
\nc{\Aram}{{\operatorname{\sf A}}}

\makeatletter
\def\grd@save@target#1{%
  \def\grd@target{#1}}
\def\grd@save@start#1{%
  \def\grd@start{#1}}
\tikzset{
  grid with coordinates/.style={
    to path={%
      \pgfextra{%
        \edef\grd@@target{(\tikztotarget)}%
        \tikz@scan@one@point\grd@save@target\grd@@target\relax
        \edef\grd@@start{(\tikztostart)}%
        \tikz@scan@one@point\grd@save@start\grd@@start\relax
        \draw[minor help lines,magenta] (\tikztostart) grid (\tikztotarget);
        \draw[major help lines] (\tikztostart) grid (\tikztotarget);
        \grd@start
        \pgfmathsetmacro{\grd@xa}{\the\pgf@x/1cm}
        \pgfmathsetmacro{\grd@ya}{\the\pgf@y/1cm}
        \grd@target
        \pgfmathsetmacro{\grd@xb}{\the\pgf@x/1cm}
        \pgfmathsetmacro{\grd@yb}{\the\pgf@y/1cm}
        \pgfmathsetmacro{\grd@xc}{\grd@xa + \pgfkeysvalueof{/tikz/grid with coordinates/major step}}
        \pgfmathsetmacro{\grd@yc}{\grd@ya + \pgfkeysvalueof{/tikz/grid with coordinates/major step}}
        \foreach \x in {\grd@xa,\grd@xc,...,\grd@xb}
        \node[anchor=north] at (\x,\grd@ya) {\pgfmathprintnumber{\x}};
        \foreach \y in {\grd@ya,\grd@yc,...,\grd@yb}
        \node[anchor=east] at (\grd@xa,\y) {\pgfmathprintnumber{\y}};
      }
    }
  },
  minor help lines/.style={
    help lines,
    step=\pgfkeysvalueof{/tikz/grid with coordinates/minor step}
  },
  major help lines/.style={
    help lines,
    line width=\pgfkeysvalueof{/tikz/grid with coordinates/major line width},
    step=\pgfkeysvalueof{/tikz/grid with coordinates/major step}
  },
  grid with coordinates/.cd,
  minor step/.initial=.2,
  major step/.initial=1,
  major line width/.initial=2pt,
}
\makeatother

\tikzset{
  treenode/.style = {align=center, inner sep=0pt, text centered,
    font=\sffamily},
  arn_n/.style = {treenode, circle, white, font=\sffamily\bfseries, draw=black,
    fill=black, text width=1.5em},
  arn_r/.style = {treenode, circle, red, draw=red, 
    text width=1.5em, very thick},
  arn_x/.style = {treenode, rectangle, draw=black,
    minimum width=0.5em, minimum height=0.5em}
}

\begin{document}

\title{Activated zero-error classical communication over quantum channels assisted with quantum no-signalling correlations}
\date{\today}
\author{Runyao Duan}
\email{runyao.duan@uts.edu.au}
\homepage{http://www.uts.edu.au/staff/runyao.duan}
\affiliation{Centre for Quantum Software and Information, Faculty of Engineering and Information Technologies, \\University of Technology Sydney, Australia}
\affiliation{UTS-AMSS Joint Research Laboratory for Quantum Computation and Quantum Information Processing, Academy of Mathematics and Systems Science, Chinese Academy of Sciences, China}

\author{Xin Wang}
\email{xin.wang-8@student.uts.edu.au}
\homepage{http://xinwang.online}

\affiliation{Centre for Quantum Software and Information, Faculty of Engineering and Information Technologies, \\University of Technology Sydney, Australia}
\maketitle

\begin{abstract}
We study the \textit{activated quantum no-signalling-assisted zero-error classical capacity} by first allowing the assistance from some noiseless forward communication channel and later paying back the cost of the helper.  This activated communication model considers the additional forward noiseless channel as a catalyst for communication.
First, we show that the one-shot activated capacity can be formulated as a semidefinite program and we derive a number of striking properties of this capacity. We further present a sufficient condition under which a noisy channel can be activated.
Second, we find that one-bit noiseless classical communication is able to fully activate any classical-quantum channel to achieve its asymptotic capacity, or the semidefinite (or fractional) packing number. Third, we prove that the asymptotic activated capacity cannot exceed the original asymptotic capacity of any quantum channel. We also show that the asymptotic no-signalling-assisted zero-error capacity does not equal to the semidefinite packing number for quantum channels, which differs from the case of classical-quantum channels. 
\end{abstract}

\section{Introduction}

A fundamental problem of information theory is to determine the capacity of a communication channel, which describes the capability of the channel for delivering information from the sender to the receiver. Shannon first discussed this problem in the zero-error setting and described the zero-error capacity of a channel as the maximum rate at which it can be used to transmit information perfectly \cite{Shannon1956}. It is well-known that the Shannon zero-error capacity is extremely difficult to compute even for very simple classical channels. Nevertheless, this capacity is upper bounded by the  Lov\'asz $\vartheta$ function \cite{Lovasz1979} which is efficiently computable by semidefinite programming \cite{Vandenberghe1996}.

Recently the zero-error information theory has been studied in the quantum setting and many interesting phenomena were observed. For instance, it was shown that shared entanglement could sometimes improve the zero-error capacity of a classical channel \cite{Cubitt2010, Leung2012},
and the entanglement-assisted zero-error channel coding was studied in \cite{Piovesan2015,Briet2015,Stahlke2014}. Furthermore, both the zero-error classical and quantum capacities can be super-activated \cite{Duan2008a, Duan2009, Cubitt2011a, Cubitt2012}. Another notable fact is that the entanglement-assisted zero-error capacity of a classical channel is also upper-bounded by the Lov\'asz $\vartheta$ function \cite{Beigi2010,Duan2013}, and this result can be generalized to quantum setting by using a quantum version of Lov\'asz $\vartheta$ function \cite{Duan2013}.  Recently the separation between the quantum  Lov\'asz $\vartheta$ function and entanglement-assisted zero-error capacity of quantum channels was shown in \cite{Wang2016f}, while the case of classical channels remains unknown.

As more general resources, the no-signaling (NS) correlations have been considered to assist the zero-error communication in
\cite{Cubitt2011,Duan2016}, respectively.  
The no-signalling correlations arise in the research of the relativistic causality of quantum operations \cite{Beckman2001, Eggeling2002a, Piani2006, Oreshkov2012}. 
Cubitt et al. \cite{Cubitt2011} first introduced classical no-signalling correlations into the zero-error classical communication.  One of us and Winter \cite{Duan2016} further introduced quantum no-signalling correlations into the zero-error communication problem. They formulated the one-shot capacity as a semidefinite program (SDP) which depends only on the non-commutative bipartite graph of the given channel (an operator space that is given by the linear span of the Choi-Kraus operators of the channel).
  Furthermore, Duan, Severini, and Winter studied the zero-error communication via quantum channel assisted by unlimited noiseless feedback and showed the induced capacity also depends only on the non-commutative bipartite graph \cite{Duan2015}.  

In this paper, we further develop the theory of quantum NS-assisted communication by introducing the activated communication model. The model is introduced in Section \ref{sec:model} and it considers the additional forward noiseless channel as a catalyst for communication. For a quantum channel $\cN$, we can ``borrow'' a noiseless classical channel $\cI$, then we can use $\cN \ox \cI$ to transmit information. After the communication finishes we ``pay back'' the capacity of $\cI$.  This kind of communication method was suggested in \cite{Acin2015a}, and was highly relevant to the notion of \textit{potential capacity} recently studied by Winter and Yang \cite{Yang2015}. We further show that the one-shot activated zero-error capacity can also be formulated as an SDP.
In Section \ref{sec: cq},
 we show a striking result that one bit can even fully activate any cq channel to achieve its asymptotic NS-assisted zero-error capacity (or the fractional packing number). 
In Section \ref{sec:asy}, we further show that there is no activation in the asymptotic regime and the one-shot activated capacity is better than the super-dense coding bound in \cite{Duan2016}.

\section{Zero-error communication over quantum channels}
A quantum channel $\cN_{A'\to B}$ is a completely positive (CP) and trace-preserving (TP) linear map from operators on a finite-dimensional Hilbert space $A'$ to operators on a finite-dimensional Hilbert space $B$. The Choi-Jamio\l{}kowski matrix of $\cN$ is  $J_{\cN}=\sum_{ij} \ketbra{i}{j}_A \ox \cN(\ketbra{i}{j}_{A'})=(\text{id}_{A}\ox\cN)\proj{\Phi}$, where $A$ and $A'$ are isomorphic Hilbert spaces with respective orthonormal basis $\{\ket i_A\}$ and $\{\ket j_{A'}\}$, $\ket \Phi =\sum_i \ket{i}_A\ket{i}_{A'}$, and $\text{id}_{A}$ is the identity map. 

In the communication task,
Alice wants to send the classical messages to Bob using the composite channel $\cM_{A\to B'}=\Pi_{AB\to A'B'}\circ\cN_{A'\to B}$, where $\Pi$ is a quantum bipartite operation that generalizes the usual encoding scheme $\cE$ and decoding scheme $\cD$ (see Figure.~\ref{fig:codes}).
We say such $\Pi$ is an NS-assisted code if it can be implemented by local operations with quantum no-signalling correlations. Here, the no-signalling constraint means that Alice and Bob cannot use the bipartite operation $\Pi$ to communicate classical information. The NS-assisted codes have also been applied to study the ordinary classical and quantum communication over quantum channels  (e.g., \cite{Duan2016,Leung2015c,Matthews2012,Wang2017d, Lai2015,Wang2016b,Park2016,Wang2016a,Matthews2016,Xie2017a}).

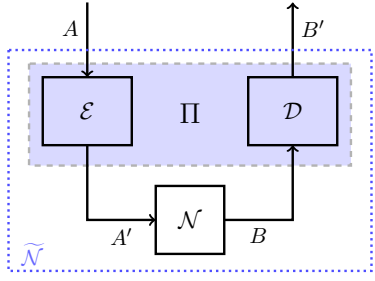
\begin{figure}[h]
\centering
\resizebox {5cm} {!} {
\begin{tikzpicture}[scale = 1.2]
    \def\xbb{0.5};\def\xb{1.5};\def\xsh{1.3};\def\ysh{1};
    \def\ya{1.3};\def\yc{3.4};\def\lo{0.2};\def\loo{0.4};
    \pgfmathsetmacro\xa{-\xb};\pgfmathsetmacro\xaa{-\xbb};
    \pgfmathsetmacro\xc{\xa-\xsh/2};\pgfmathsetmacro\xd{\xa+\xsh/2};
    \pgfmathsetmacro\xe{\xb-\xsh/2};\pgfmathsetmacro\xf{\xb+\xsh/2};
    \pgfmathsetmacro\yb{\ya+\ysh};
    \draw[very thick,dashed,fill=blue!50,opacity=0.3] (\xc-\lo,\yb+\lo) rectangle (\xf+\lo,\ya-0.3);
    \draw[very thick,->] (\xa,\yc) -- node[left,shift={(0,0.2)}] {$A$} (\xa,\yb);
    \draw[very thick,<-] (\xb,\yc) -- node[right,shift={(0,0.2)}] {$B'$} (\xb,\yb);
    \draw[very thick] (\xc,\yb) rectangle (\xd,\ya) node[midway] {\large $\cE$};
    \draw[very thick] (\xe,\yb) rectangle (\xf,\ya) node[midway] {\large $\cD$};
    \draw[very thick,->] (\xa,\ya) -- (\xa,0.2) -- node[below] {$A'$} (\xaa,0.2);
    \draw[very thick,<-] (\xb,\ya) -- (\xb,0.2) -- node[below] {$B$} (\xbb,0.2);
    \draw[very thick] (\xaa,\ysh/2+0.2) rectangle (\xbb,-\ysh/2+0.2) node[midway] {\large $\cN$};
    \node[] (dots) at (0,\yb-0.52) {\Large $\Pi$};
    \draw[very thick,dotted,blue!70] (\xc-\loo-0.1,\yb+\loo) rectangle (\xf+\loo+0.1,-\ysh/2-\lo+0.15);
    \node[blue!60] (dots) at (\xc-0.15,-\ysh/2+0.2) {$\widetilde \cN$};
\end{tikzpicture}
}
\caption{General code scheme}
\label{fig:codes}
\end{figure}

We denote $\cM_{\rm{0,NS}}(\cN)$ as the maximum number of bits can be transmitted perfectly over a single use of quantum channel $\cN$ with NS-assisted codes, i.e.,
\begin{equation}
\cM_{\rm{0,NS}}(\cN):=\sup\{\log \ell:  \Pi\circ\cN= \cI_\ell, \Pi \text { is an NS-assisted code} \},
\end{equation}
where
 $\cI_\ell(\rho)=\sum_{i=0}^{\ell-1}\tr(\rho\proj{i})\proj{i}$ is the classical noiseless channel.
Also, throughout this paper, $\log$ denotes the binary logarithm $\log_2$. 

  As showed in \cite{Duan2016}, for a channel  $\cN(\rho)=\sum_{k}E_k\rho E_k^{\dagger}$, the one-shot NS-assisted zero-error $\cM_{\rm{0,NS}}(\cN)$ only depends on the non-commutative bipartite graph $K$ of the channel, i.e.,
\begin{equation}
\cM_{\rm{0,NS}}(\cN)=\cM_{\rm{0,NS}}(K),
\end{equation}
where  $K=\rm{span}\{E_k\}$ is the linear span of the Choi-Kraus operators of the channel.
For a quantum channel $\cN$ with non-commutative bipartite graph $K$, the one-shot NS-assisted zero-error capacity is gievn by \cite{Duan2016}
 \begin{equation}
 \cM_{\rm{0,NS}}(\cN)= \cM_{\rm{0,NS}}(K)=\log \U(K), 
\end{equation}
 where $\U(K)$ is given by the following SDP:
\begin{equation}\begin{split}\label{eq:Upsilon}
\U(K) = \max &\tr S_A \\
 \text{ s.t. }& 0 \leq U_{AB} \leq S_A \ox \1_B, \\
        &  \tr_A U_{AB} = \1_B, \\
        &\tr P_{AB}(S_A\ox\1_B-U_{AB}) = 0.
\end{split}\end{equation}
Here, $P_{AB}$ denotes the projection onto the support of the Choi-Jamio\l{}kowski matrix of $\cN$, which is uniquely determined by the non-commutative bipartite graph $K$. Due to this fact, we will not distinguish between the notations $\U(\cN)$ and ${\U}(K)$ in the rest of this paper. 

Then by the regularization, the NS-assisted zero-error capacity is
\begin{equation}\label{C0NS}
C_{0,\NS}(\cN)    = \sup_{n\geq 1} \frac1n \log \U\left(K^{\ox n}\right).
\end{equation}
The $\sup$ in Eq. (\ref{C0NS}) can be replaced by $\lim$ based on  the lemma about existence of limits in~\cite{Barnum}.

\section{One-shot activated zero-error communication}\label{sec:model}
\subsection{Activated one-shot zero-error capacity}
The model of activated communication is described as follows. For a quantum channel $\cN$ assisted by NS codes, we can first borrow a noiseless classical channel $\cI_\ell$ whose capacity is $\log \ell$, then we can use $\cN\ox \cI_\ell$ coherently to transmit classical messages. After the communication finishes, we just pay back the capacity of $\cI_\ell$ (see Figure~\ref{fig: model}.) The communication model follows the idea of potential capacities of quantum channels introduced by Winter and Yang \cite{Yang2015}.
\begin{figure}[h]
\centering
  \begin{tikzpicture}[scale = 0.8]
 
 \node[very thick] at (-1.5,1) {$m\in\{1,\dots,M\}$};
 \draw[very thick,->] (0,1) -- node[above] {} (1,1);
  \draw[very thick,->] (6,1) -- node[above] {} (7,1);
   \node[very thick] at (8.8,1) {$\hat m\in\{1,\dots,M\}$};
  
    \def\xa{1.7};\def\xb{2.5};\def\xc{3};\def\xd{4};\def\xe{5};\def\xf{6};\def\xg{7};
    \def\ya{0.8};\def\yc{-1.6};\def\yd{-2.4};  \def\o{0.2} \def\ye{1}; \def\yf{2};
    \pgfmathsetmacro\yb{-\ya}
    \draw[very thick,-] (2,\yf) -- node[above] {} (3,\yf);
        \draw[very thick,-] (4,\yf) -- node[above] {} (5,\yf);
            \draw[very thick,-] (2,0) -- node[above] {} (3,0);
        \draw[very thick,-] (4,0) -- node[above] {} (5,0);

   \draw[very thick]  (1,2.5) rectangle (2,-0.6) node[midway] {$\cE$};
    \draw[very thick]  (\xc,\ya-0.2) rectangle (\xd,\yb+0.2) node[midway] {$\cN$};
      \draw [very thick, dashed] (\xc,\ya+1.6) rectangle (\xd,\yb+2.2) node[midway] {$\cI_\ell$};
    \draw[very thick]  (5,2.5) rectangle (\xf,-0.6) node[midway] {$\cD$};
%
\end{tikzpicture}
\caption{Activated classical communication.}
\label{fig: model}
\end{figure}
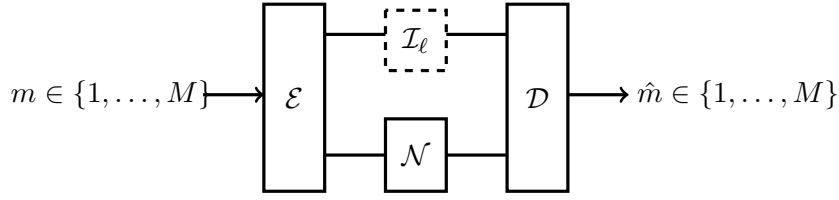

\begin{definition} 
For a quantum channel $\cN$ with non-commutative bipartite graph $K$, the one-shot activated no-signalling assisted zero-error classical capacity is defined as the following:
\begin{equation}
\cM_{\rm{0,NS}}^a(\cN)=\cM_{\rm{0,NS}}^a(K):=\mathop {\sup }\limits_{\ell \ge 1} [\cM_{\rm{0,NS}}(K \otimes \D_\ell) - \log {\ell}],
\end{equation}
where $\D_\ell$ is the non-commutative graph of the noiseless channel $\cI_\ell(\rho)=\sum_{i=0}^{\ell-1}\tr(\rho\proj{i})\proj{i}$.
\end{definition}


\begin{definition} 
For a quantum channel $\cN$ with non-commutative bipartite graph $K$,
the asymptotic activated no-signalling zero-error classical capacity is given the following regularization:
\begin{equation}
C_{0,\NS}^a(\cN)=C_{0,\NS}^a(K):=\sup_{n\geq 1} \frac1n \cM_{\rm{0,NS}}^a\left(K^{\ox n}\right).
\end{equation}
\end{definition}


To provide a feasible formulation of the activated capacity $\cM_{\rm{0,NS}}^a(\cN)$, let us frsit introduce a slightly revised SDP of $\U(K)$ as follows,
  \begin{equation}\begin{split}
    \label{eq:Upsilon-hat}
  \widehat{\U}(K)
          = \max &\tr S_A \\
    \text{ s.t. }& 0 \leq U_{AB} \leq S_A \ox \1_B, \\
          & \tr_A U_{AB} \leq \1_B,    \\
          & \tr P_{AB}(S_A\ox\1_B-U_{AB}) = 0.
  \end{split}\end{equation}
 The only difference between $\widehat{\U}(K)$ and $\U(K)$ is that now $\tr_A U_{AB}$ is only required to be less than or equal to $\1_B$, and an equality is not necessary. However, we will see that such a small revision is of crucial importance. The dual SDP of  $\widehat{\U}(K)$ is given by
\begin{equation}\begin{split}
  \label{eq:Upsilon-hat-dual}
  \widehat{\U}(K) = \min &\tr T_B \\
   \text{ s.t. } & V_{AB} \leq \1_A \ox T_B,    \\
        &\tr_B V_{AB} \geq \1_A,\ T \geq 0, \\
        & (\1-P)_{AB}V_{AB}(\1-P)_{AB} \leq 0.
\end{split}\end{equation}
Note that by strong duality, the values of both the primal and the dual SDPs coincide.

Now we are ready to present the main result.
\begin{theorem}
For any quantum channel $\cN$ with non-commutative bipartite graph $K$,  
\begin{equation}\begin{split}
\cM_{\rm{0,NS}}^a(\cN)=\log \widehat\U(K).
  \end{split}\end{equation} 
\end{theorem}

\begin{proof} 
The intuition of this theorem is that the additional noiseless channel may play a role as a catalyst during the communication task.

To prove the achievable part, it's important to observe that the additional noiseless channel indeed provides a larger solution space of $\U(K\ox \D_\ell)$.
Let us first consider the case $\ell=2$ and assume that the optimal feasible solution of $\widehat{\U}(K)$ is $\{S_A, U_{AB}\}$. 
Let us choose 
\begin{align}
S_{AA'}=S_A \ox (\ket {0} \bra {0}+\ket{1}\bra{1})_{A'}
\end{align}
 and
\begin{equation}\begin{split}
U_{AA'BB'}=&U_{AB} \ox (\ket {00} \bra {00}+\ket{11}\bra{11})_{A'B'}
+\bar U_{AB} \ox (\ket {01} \bra{01}+\ket{10}\bra{10})_{A'B'},
\end{split}\end{equation}
where $\bar U_{AB}=\frac{S_A}{\tr S_A} \ox (\1_B-\tr_A U_{AB}).$

This construction ensures that
\begin{equation}
\tr_{AA'}U_{AA'BB'}=\tr_{A}((U_{AB} +\bar U_{AB}) \ox \1_{B'})=\1_{BB'}.
\end{equation}

Moreover, we have
\begin{align}
S_{AA'}\ox\1_{BB'}-U_{AA'BB'}
=&(S_{A}\ox \1_B-U_{AB}) \ox(\ket {00} \bra {00}
+\ket{11}\bra{11})_{A'B'}\\
&+(S_{A}\ox \1_B-\bar U_{AB}) \ox(\ket {01} \bra {01}+\ket{10}\bra{10})_{A'B'},
\end{align}
which directly means that
\begin{equation}
S_{AA'}\ox\1_{BB'}-U_{AA'BB'} \ge 0.
\end{equation}
Also, noting that $\tr P_{AB}(S_{A}\ox \1_B-U_{AB}) =0$ and $\tr \1_{A'B'}(\ket {01} \bra {01}+\ket{10}\bra{10})_{A'B'}=0$, we have that
\begin{equation}
\tr(P_{AB}\ox\1_{A'B'})(S_{AA'}\ox\1_{BB'}-U_{AA'BB'})=0.
\end{equation}

Therefore, $\{S_{AA'}, U_{AA'BB'}\}$ is a feasible solution of $\U(K \otimes {\D_2})$, which means that
\begin{equation}\label{eq:achieve}
\mathop {\sup }\limits_{\ell \ge 2} \frac{{\U(K \otimes \D_\ell})}{\ell} \ge \frac{\U(K \otimes {\D_2})}{2}\ge \frac{\tr S_{AA'}}{2}=\widehat \U(K).
\end{equation}


On the other hand, to prove the converse part, we will use the fact that $\widehat{\U}(K\ox\Delta_\ell) = \ell\,\widehat{\U}(K)$, which is provided in the following Lemma~\ref{lemma:converse}.
This fact directly implies that
\begin{equation}\label{eq:converse}
 \mathop {\sup }\limits_{\ell \ge 2} \frac{{\U(K \otimes \D_\ell})}{\ell} \le  \mathop {\sup }\limits_{\ell \ge 2} \frac{{\widehat\U(K \otimes \D_\ell})}{\ell} =\widehat \U(K).
\end{equation}

Finally, by Eq.~\eqref{eq:achieve} and Eq.~\eqref{eq:converse}, we can conclude that
\begin{equation}
\cM_{\rm{0,NS}}^a(\cN)=\cM_{\rm{0,NS}}^a(K)=\log \widehat\U(K).
\end{equation}
\end{proof}

A simple but useful property of $\widehat{\U}$ is shown as follows.
\begin{lemma}\label{lemma:converse}
For any non-commutative bipartite graph $K$, we have
$$\widehat{\U}(K\ox\Delta_\ell) = \ell\,\widehat{\U}(K).$$
\end{lemma}
\begin{proof}
On one hand, it is evident from the super-multiplicativity that $\widehat{\U}(K\ox\Delta_\ell) \ge \ell\,\widehat{\U}(K)$.
On the other hand, note that an optimal solution for SDP (\ref{eq:Upsilon-hat-dual}) for $\Delta_\ell$ is given by
$\{\1_{B'}, \sum_{i=1}^\ell \proj{ii}_{A'B'}\}$,
and we assume that the optimal solution of SDP (\ref{eq:Upsilon-hat-dual}) for $K$ is $\{T_{B},V_{AB}\}$.
It is evident that 
\begin{align}
V_{AB}\ox\sum_{i=1}^\ell \proj{ii}_{A'B'} \leq \1_{AA'} \ox T_{B}\ox\1_{B'}.
\end{align}
Then, it can be checked that $\{V_{AB}\ox \sum_{i=1}^\ell \proj{ii}_{A'B'}$, $T_{B}\ox\1_{B'}\}$
is a feasible solution of SDP(\ref{eq:Upsilon-hat-dual}) for $\widehat{\U}(K\ox\Delta_\ell)$. Therefore,$\widehat{\U}(K\ox\Delta_\ell) \le \tr T_{B}\ox\1_{B'} =\ell\,\widehat{\U}(K)$.
\end{proof}

Indeed, in the previous proof we have shown the following stronger result:
$$\widehat \U(K)=\mathop \frac{{\U(K \otimes \D_\ell})}{\ell}=\frac{{\U(K \otimes {\D_{2}})}}{2}, \forall~\ell\geq 2.$$

\subsection{Activation via noisy quantum channels}
We further discussion the activation via noisy quantum channels
\begin{proposition}\label{upsilon two}
Let us consider two quantum channels $\cN_1$ with non-commutative bipartite graphs $K_1$ and $K_2$, respectively. If $\U(K_2)-1\ge\frac{1}{\widehat \U(K_1)}$, then
\begin{equation}
\cM_{\rm{0,NS}}(K_1 \ox K_2) - \cM_{\rm{0,NS}}(K_2) \ge \cM_{\rm{0,NS}}^a(K_1).
\end{equation}
 In other words, $K_2$ can activate $K_1$ if $K_1$ is activatable.
 In particular, this inequality always holds when $\U(K_2) \ge 2$.
\end{proposition}
\begin{proof}
Let us assume that the optimal solution to the SDP~\eqref{eq:Upsilon-hat} of $\widehat\U(K_1)$ is $\{S_A, U_{AB}\}$  while
the optimal solution to the SDP~\eqref{eq:Upsilon} of $\U(K_2)$
 is $\{S_{A'}, U_{A'B'}\}$.
 
Then we can choose 
\begin{align}
S_{AA'}&=S_A \ox S_{A'},\\
U_{AA'BB'}&=U_{AB} \ox U_{A'B'}+ \bar U_{AB} \ox V_{A'B'}, 
\end{align}
where
$V_{A'B'}=({S_{A'}\ox\1_{B'}- U_{A'B'}})/({\tr S_{A'}-1})$  and $\bar U_{AB}={S_A}/{\tr S_A} \ox (\1_B-\tr_A U_{AB})$.
Furthermore, we have
\begin{equation}
S_{AA'} \ox \1_{BB'}-U_{AA'BB'}=(S_A \ox \1_B-U_{AB}) \ox U_{A'B'}+[(\tr S_{A'}-1-\frac{1}{\tr S_A})S_A \ox \1_B] \ox V_{A'B'}.
\end{equation}
Then, one can check that the constucted solutions satisfy 
\begin{align}
S_{AA'}\ox\1_{BB'}-U_{AA'BB'} \ge 0,
\tr_{AA'}U_{AA'BB'}=\1_{BB'},\\
(P_{AB}\ox P_{A'B'})(S_{AA'}\ox\1_{BB'}-U_{AA'BB'} )= 0.
\end{align}

Therefore, $\{S_{AA'}, U_{AA'BB'}\}$ is a feasible solution to  the SDP~\eqref{eq:Upsilon} of $\U(K_1 \ox K_2)$, which means that
\begin{align}
\U(K_1 \ox K_2) \ge \widehat \U(K_1) \U(K_2).
\end{align}
\end{proof}

If we only consider using the channel $\cN$ to activate itself, we have the following result from the above proposition.

For any quantum channel $\cN$ with non-commutative bipartite graph $K$, if $\U(K) \ge \frac{1+\sqrt5}{2}$, then
\begin{equation}
\frac{\U(K \ox K)}  {\U(K)}\ge \widehat \U(K).
\end{equation}
Note that $\U(K) \ge \frac{1+\sqrt5}{2}$ means $ \U(K)-1-\frac{1}{\U(K)} \ge 0$. Thus the result follows directly from Proposition \ref{upsilon two}.

\section{Activated zero-error capacity of classical-quantum channel}\label{sec: cq}
A classical-quantum (cq) channel $\cN:i\to \rho_i (1\le i\le n)$ is a CPTP map with classical inputs $\{i\}_{i=1}^n$ and quantum outputs  $\{\rho_i\}_{i=1}^n$.
The non-commutative bipartite graph of a cq channel is called a cq graph.  Given a cq channel  $\cN:i\to \rho_i  (1\le i\le n)$ with cq graph $K$, its one-shot NS-assisted zero-error capacity (quantified as messages) can be simplified to
\begin{equation}\begin{split}
  \label{eq:Upsilon:cq channel}
  \U(K) = \max & \sum_i s_i \\
   \text{ s.t. }&\  0 \leq s_i,\ 0 \leq R_i \leq s_i(\1-P_i),\\
        &\sum_i (s_i P_i + R_i) = \1.
\end{split}\end{equation}
where $P_i$ is the projection onto the support of $\rho_i$ for $1\le i \le n$

Moreover,   it was shown in \cite{Duan2016} that the asymptotic no-signalling assisted  zero-error classical capacity of a cq channel is equal to the semidefinite (fractional) packing number first suggested by Harrow \cite{Duan2016}.
\begin{lemma}(Theorem 4 in \cite{Duan2016})
For any  classical-quantum channel $\cN:i\to \rho_i$ $(1\le i\le n)$ with cq graph $K$, 
\begin{equation}
C_{\rm{0,NS}}(\cN)=\log \Aram(K),
\end{equation}
with
\begin{equation}\begin{split}
  \label{eq:Aram:cq channel}
  \Aram(K) = \max &\sum_i s_i \\ 
  \text{ s.t. }&  0 \leq s_i,\ \sum_i s_i P_i \leq \1.
  \end{split}
\end{equation}
where $P_i$ the projection onto the support of $\rho_i$ for $1\le i\le n$.
\end{lemma}
This result is a classical-quantum generalization of the fact that the fractional packing/covering number~\cite{Shannon1956,Scheinerman2011} of the bipartite graph (induced by the classical channel) equals to its NS-assisted zero-error capacity~\cite{Cubitt2011}. Moreover, Shannon proved that the feedback-assisted zero-error capacity of a classical channel is also given by the fractional packing number~\cite{Shannon1956}.

For any channel $\cN$ with cq graph $K$, the one-shot activated capacity  $\cM_{\rm{0,NS}}^a(\cN)=\log \widehat \U(K)$ can be simplied to
\begin{equation}\begin{split}
  \label{eq:hat-Upsilon:cq channel}
 \widehat \U(K)= \max  &\sum_i s_i \\
\text{ s.t. }&  0 \leq s_i,\ 0 \leq R_i \leq s_i(\1-P_i),\\
        & \sum_i (s_i P_i + R_i) \le \1. \
\end{split}\end{equation}

\begin{theorem}\label{activated-cq}
For any classical-quantum channel $\cN$ with cq graph $K$, 
\begin{equation}
M_{\rm{0,NS}}^a(\cN)=\log\Aram(K). 
\end{equation}
In other words, for any cq channel, the asymptotic NS-assisted zero-error capacity (or the semidefinite packing number) can be achieved via activated NS codes in the one-shot regime, i.e.,
 \begin{equation}
C_{\rm{0,NS}}^a(\cN)=M_{\rm{0,NS}}^a(\cN)=\log\Aram(K).
\end{equation}
\end{theorem}
\begin{proof}
First, we will show  $\Aram(K)  \ge \widehat \U(K)$. Suppose that optimal solution of the SDP~\eqref{eq:hat-Upsilon:cq channel} of $\widehat \U(K)$ is $\{s_i,R_i\}$. Then, 
\begin{align}
\sum_i s_iP_i \le \1-\sum_i R_i\le \1,
\end{align}
which means that $\{s_i\}$ is a feasible solution for $\Aram(K)$. So we have $\Aram(K)  \ge \widehat \U(K)$.

Second, let us assume the optimal solution of SDP~(\ref{eq:Aram:cq channel}) is $\{s_i\}$, let $R_i=0$ for all $i$. It is easy to check that $\{s_i,R_i\}$ is a feasible solution of SDP~(\ref{eq:hat-Upsilon:cq channel}), which means that $\Aram(K)  \le \widehat \U(K)$.
Therefore, for any cq graph $K$, it holds that 
\begin{equation}
\widehat \U(K)=\Aram(K).
\end{equation}
\end{proof}

To see the existence of activation, let us consider an example here.

\vspace{0.25cm}
\begin{example}
For the simplest possible cq channel $\cN$, which has only two inputs and two pure output states $P_i = \proj{\psi_i}$.
Withour loss of generarity, we assume that $\ket{\psi_0} = \alpha\ket{0} + \beta\ket{1}$  and
  $\ket{\psi_1} = \alpha\ket{0} - \beta\ket{1}$ with $\alpha \geq \beta = \sqrt{1-\alpha^2}$.
  In \cite{Duan2016}, it has been solved that
$\U(K)=1$ and  $\Aram(K) =  \frac{1}{\alpha^2} $.
Hence, by Theorem \ref{activated-cq}, we know
\begin{equation}
\widehat \U(K)=\frac{\U(\cN \otimes {\D_2})}{2} = \frac{1}{\alpha^2}>\U(K)=1.
\end{equation}
Furthermore, we have
 \begin{equation}
C_{\rm{0,NS}}^a(\cN)=M_{\rm{0,NS}}^a(\cN)=-2\log \a >M_{\rm{0,NS}}(\cN)=0.
\end{equation}
\end{example}

\section{Activated zero-error communication in the asymptotic regime}\label{sec:asy}
\subsection{Asymptotic zero-error capacity}
As we find the activation phenomenon of zero-error communication in the one-shot regime, it's natural to wonder whether there exists an activation in the asymptotic regime. In the following theorem, we prove that the answer is nagetive.
\begin{theorem}\label{asymptotic activated capacity}
 For any channel $\cN$ with non-commutative bipartite graph K with positive zero-error capacity, let $n_0$ be the smallest integer such that $\U(K^{\ox n_0})\ge2$. Note that $n_0$ always exists and depends only on K. Then for any $n\ge n_0$, we have
\begin{equation}\label{asymptotic hatU=U}
2\le \widehat{\U}(K^{\ox {(n-n_0)}}) \le \U(K^{\ox n}) \le \widehat{\U}(K^{\ox n}).
\end{equation}
Moreover,
\begin{equation}\label{asymptotic hatU}
C_{\rm{0,NS}}^a(K)= \mathop {\sup }\limits_{n \ge 1} \log \sqrt[n]{{\widehat \U({K^{ \otimes n}})}} = \mathop {\lim }\limits_{n \to \infty }\log \sqrt[n]{{\widehat \U({K^{ \otimes n}})}}=C_{\rm{0,NS}}(K).\end{equation}
\end{theorem}
\begin{proof}
Eq. (\ref{asymptotic hatU=U}) is immediatelly from Theorem \ref{upsilon two}. Then,
\begin{equation}
\mathop {\lim }\limits_{n \to \infty }\log \sqrt[n]{{\widehat \U({K^{ \otimes n}})}}=\mathop {\lim }\limits_{n \to \infty }\log \sqrt[n]{{ \U({K^{ \otimes n}})}}.
\end{equation}

To prove Eq. (\ref{asymptotic hatU}), the technique is based on a lemma about the existence of limits in \cite{Barnum}.
On one hand,  $\log \widehat \U({K^{ \otimes n}}) \le 2n\log d$.
On the other hand, since $\widehat \U(K)$ is super-multiplicative, then
$\log \widehat \U({K^{ \otimes (m  n)}}) \ge \log \widehat \U({K^{ \otimes m}}) + \log  \widehat \U({K^{ \otimes n}})$.
Therefore,
\begin{equation}
\mathop {\sup }\limits_{n \ge 1} \frac{\log \widehat \U({K^{ \otimes n}})}{n} =\mathop {\lim }\limits_{n \to \infty }\frac{\log \widehat \U({K^{ \otimes n}})}{n}=C_{\rm{0,NS}}(K).
\end{equation}
\end{proof}

\subsection{Lower bound of the activated capacity}
We further explore the lower bound for the activated zero-error communication assisted with NS codes and study the equivalent conditions for a quantum channel to have a positive capacity.  We find that the one-shot activated capacity is always larger than or equal to the super-dense coding bound in  Theorem 25 in \cite{Duan2016}.
\begin{proposition}
\label{feasibility-K}
Let $K$ be a non-commutative bipartite graph with Choi-Jamio\l{}kowski
projection $P_{AB}$, and let $Q_{AB}=\1_{AB}-P_{AB}$ be the orthogonal
complement of $P_{AB}$. Let $P_B=\tr_A P_{AB}$, then the following are equivalent:
\renewcommand{\theenumi}{\roman{enumi}}
\begin{enumerate}
  \item $C_{0,{\rm NS}}(K)>0$;
 \item $\Aram(K)>1$;
  \item $P_{B}<d_A\1_B$;
  \item $\tr_A Q_{AB}$ is positive definite;
 \item $\cM_{\rm{0,NS}}^a(K)>0$ (or $\widehat \U(K) >1$).
\end{enumerate}
\renewcommand{\theenumi}{\arabic{enumi}}

As a matter of fact, we have
\begin{equation}\begin{split}
 &C_{\rm{0,NS}}(\cN) \geq \cM_{0,{\rm NS}}^a(\cN)\geq \log \frac{d_A}{\|\tr_A P_{AB}\|_\infty}.
\end{split}\end{equation}
\end{proposition}
\begin{proof}
In \cite{Duan2016}, i) , ii), iii) and iv) are proved to be equivalent. We focus on v) here.

On one hand, it is easy from v) to i) by Theorem \ref{asymptotic activated capacity}.

On the other hand, to prove  i) to v), it is equivalent to prove  iii) to v).
We first apply the standard super-dense coding protocol to obtain a cq channel $\cN_s$ with $d_A^2$ outputs $\{(U_m\ox \1_B) J_{AB} (U_m\ox \1_B)^\dag\}$, and the projections are given by $\{(U_m\ox \1_B)P_{AB} (U_m\ox \1_B)^\dag\}$, where $U_m$ are generalized Pauli matrices acting on $A$. Therefore,
\begin{equation}\begin{split}\label{superdense-bound}
\Aram(\cN_s)=&\frac{d_A^2}{\sum_{m=1}^{d_A^2}(U_m\ox \1_B)P_{AB} (U_m\ox \1_B)^\dag }\\
&=\frac{d_A}{\|\tr_A P_{AB}\|_\infty}.
\end{split}\end{equation}
Hence,   $\widehat \U(K) \ge \widehat \U(\cN_s)=\Aram(\cN_s)$. And when $\tr_A P_{AB}<d_A\1_B$ strictly holds,  the right-hand side of the above equation
is strictly larger than $1$. This also means that 1 bit is enough to activate any one-shot useless non-trival channel.
\end{proof}

\vspace{0.25cm}

\begin{example}
Let us consider the amplitude damping channel $\cN_r(\rho)=\sum_{i=0}^1 E_i\rho E_i^\dag$ 
with 
\begin{equation}
E_0=\ketbra{0}{0}+\sqrt{1-r}\ketbra{1}{1}, E_1=\sqrt{r}\ketbra{0}{1} (0\leq r\leq 1).
\end{equation}
The classical communication capability of this class of channels has been studied in \cite{Giovannetti2005,Brandao2011c,Wang2016g}, but its classical capacity, zero-error capacity, NS-assisted zero-error capacity are all unknown.
In particular, when $r=3/4$,
let $S_A=\frac{3}{8}\ketbra 0 0+\frac{3}{4}\ketbra 1 1$ and $U_{AB}=\frac{1}{4}(\ketbra {00} {00}+ \ketbra {00} {11}+\ketbra {11} {00}+\ketbra {11} {11})+\frac{3}{4}\ketbra {10} {10}$, it is not difficult to check that $\{S_A, U_{AB}\}$ is an feasible solution to the SDP~(\ref{eq:Upsilon-hat}) of $\U(\cN_r)$. Hence, 
\begin{align}
\cM_{\rm{0,NS}}^a(\cN_r) \ge \log \frac{\tr S_A}{2}=\log \frac{9}{8}>\cM_{\rm{0,NS}}(\cN_r)= 0.
\end{align}
By applying the super-dense coding bound \cite{Duan2016}, we know 
$C_{0,\NS}(\cN_r)\geq \log \frac{d_A}{\|\tr_A P_{AB}\|_\infty}=\log \frac{10}{9}$. 
And it's clear that the one-shot activated capacity is better than the super-dense coding bound in \cite{Duan2016}, i.e.,
\begin{align}
\cM_{\rm{0,NS}}^a(\cN_r) \ge\log \frac{9}{8}>\log \frac{10}{9}=\frac{d_A}{\|\tr_A P_{AB}\|_\infty}.
\end{align}
\end{example}

\subsection{$C_{\rm{0,NS}}$ and semidefinite packing nubmber}
As the asypmtotic capacity of classical-quantum channel is given by the semidifinite (or fractional) packing number $\Aram(K)$ in Eq.\ref{eq:Aram:cq channel}, will it hold even for general quantum channels? In \cite{Duan2016}, the semidifinite packing number for a general quantum channel was also introduced, i.e.,
\begin{equation}\begin{split}
  \label{eq:Aram}
   \Aram(\cN)= \Aram(K) = \max&\tr S_A  \\
  \text{ s.t. }&  0 \leq S_A,\tr_A P_{AB}(S_A\ox\1_B) \leq \1_B. \\
\end{split}\end{equation}

The difficulty is that we currently do not know efficient methods to calculate the asymptotic no-signalling zero-error capacity. 

 We will exhibit an example to disprove it.
\begin{proposition}
There exists a quantum channel $\cN$ with non-commutative bipartite graph $K$ such that $\widehat \U(K)>\Aram(K)$.
Consequently, 
\begin{equation}
C_{\rm{0,NS}}(\cN)\ne \Aram(\cN).
\end{equation}
\end{proposition}
\begin{proof}
Let $K$ correspond to the quantum channel $\cN(\rho)=\sum_{i=0}^{2}E_i\rho E_i^\dagger$ with $E_0=\frac{1}{\sqrt 2} \ketbra 0 0+\frac{1}{\sqrt 2} \ketbra 2 0$, $E_1=\sqrt{\frac{50}{99}}\ketbra 0 2+\sqrt{\frac{1}{99}} \ketbra 1 1+\sqrt{\frac{49}{99}}\ketbra 2 2$ and $E_2=\sqrt{\frac{98}{99}} \ketbra 0 1$. 



By solving SDPs on Matlab \cite{CVX, QETLAB}, we find that 
\begin{equation}
\widehat \U(\cN)\approx 1.1767>1.1751\ge\Aram(\cN).
\end{equation}
Then, it leads to
\begin{equation}
C_{\rm{0,NS}}(\cN)\ge \cM_{\rm{0,NS}}(\cN) >\log \Aram(\cN).
\end{equation}
\end{proof}

\section{Discussions}
%

There are many related interesting open problems, of which we highlight a few here. For example, an interesting problem to characterize the activated entanglement-assisted classical communication over quantum channels. Also, it would be interesting to study the asymptotic NS-assisted zero-error capacity, 
which is still difficult to calculate since it can be larger than the semidefinite (fractional) packing number.
 Perhaps one could consider other quantum generalizations of the fractional packing number.


This work was partly supported by the Australian Research Council (Grants No. DP120103776 and No. FT120100449) and the National Natural Science Foundation of China (Grant No. 61179030).

\end{document}